\documentclass[11pt]{article}
\usepackage[margin=1in]{geometry}
\usepackage{amsmath,amssymb,amsthm,mathtools,bm}
\usepackage{graphicx}
\usepackage{xcolor}
\usepackage{booktabs}
\usepackage{longtable}
\usepackage{url}
\usepackage{float}
\usepackage{etoolbox}
\usepackage[pdfencoding=auto,psdextra,hidelinks]{hyperref}
\usepackage{bookmark}
\newtheoremstyle{narrative}
  {6pt}
  {6pt}
  {\normalfont}
  {}
  {\bfseries}
  {.}
  {.5em}
  {}
\theoremstyle{narrative}
\newtheorem{definition}{Definition}[section]

\newtheorem{proposition}[definition]{Proposition}
\newtheorem{theorem}[definition]{Theorem}

\newtheorem{example}[definition]{Example}
\newcommand{\addqedtoenv}[1]{%
  \AtBeginEnvironment{#1}{\pushQED{\qed}}%
  \AtEndEnvironment{#1}{\popQED}%
}
\addqedtoenv{definition}
\addqedtoenv{lemma}
\addqedtoenv{proposition}
\addqedtoenv{theorem}
\addqedtoenv{corollary}
\addqedtoenv{example}
\newcommand{\F}{\mathbb{F}}
\newcommand{\vect}[1]{\bm{#1}}
\newcommand{\zero}{\bm{0}}
\newcommand{\HX}{H_X}
\newcommand{\HZ}{H_Z}
\newcommand{\HXB}{H_X^{\mathrm{base}}}
\newcommand{\HZB}{H_Z^{\mathrm{base}}}
\newcommand{\ind}[1]{\mathbf{1}_{#1}}
\newcommand{\inner}[2]{\left\langle #1,#2 \right\rangle}

\title{High-Girth Regular Quantum LDPC Codes from Affine-Coset Structures}
\author{Koki Okada \qquad Kenta Kasai\\Institute of Science Tokyo\\Email: \texttt{okada.k.3154@m.isct.ac.jp},\\\texttt{kenta@ict.eng.isct.ac.jp}}
\date{}

\begin{document}
\hypersetup{
  bookmarksnumbered=true,
  bookmarksopen=true,
  pdftitle={High-Girth Regular Quantum LDPC Codes from Affine-Coset Structures},
  pdfauthor={Koki Okada; Kenta Kasai}
}
\maketitle

\begin{abstract}
We construct a quantum low-density parity-check code family from a length-$512$ Calderbank--Shor--Steane base matrix pair. The base pair is permutation-equivalent to the known SPC(3) product CSS code, and the present affine-coset description gives a direct proof that both Tanner graphs are $(3,8)$-regular with girth $8$. The base code has parameters $[[512,174,8]]$. We then apply circulant permutation matrix (CPM) lifts. The main decoding experiment uses the CPM-lifted code with lift factor $P=32$, which has parameters $[[16384,4142,32]]$, under the code-capacity depolarizing model. A computer-assisted exhaustive search excludes nontrivial logical representatives of weight below $32$ on both sides. Together with explicit weight-$32$ non-stabilizer logical representatives, this establishes $d_X=d_Z=32$ for the fixed experimental code. A belief-propagation decoder with post-processing achieved frame error rate about $10^{-8}$ at $p=0.085$; an independently observed logical residual had weight $40$.
\end{abstract}

\section{Introduction}
Low-density parity-check (LDPC) codes were introduced as classical error-correcting codes described by sparse parity-check matrices and Tanner graphs, and their development has been closely tied to iterative decoding \cite{gallager_1962,tanner_1981}. In quantum error correction, Calderbank--Shor--Steane (CSS) codes provide a basic way to construct quantum codes from two binary check matrices \cite{calderbank_shor_1996,steane_1996}. As in classical LDPC coding, sparsity of the check matrices and control of short cycles in the Tanner graphs remain important for quantum LDPC codes. The additional condition specific to quantum CSS codes is that the two check matrices must satisfy the CSS orthogonality condition \cite{mackay_mitchison_mcfadden_2004,babar_botsinis_alanis_ng_hanzo_2015,breuckmann_eberhardt_2021}.

The guiding idea of this work is the same basic one emphasized in \cite{kasai_2026_orthogonality_barrier}: start from the design principles that make classical LDPC codes useful \cite{gallager_1962,tanner_1981,richardson_urbanke_2001} and ask how much of them can be retained in quantum error correction. In the classical setting, one can choose sparse regular parity-check matrices, avoid short cycles in the Tanner graph, and use belief-propagation (BP) decoding as the natural iterative decoder. Bringing this recipe to CSS codes is not a matter of using the same sparse graph twice, because the two check matrices must also be orthogonal. Thus the construction problem is to keep the classical advantages of regularity, sparsity, and large girth while enforcing the quantum orthogonality condition. The present paper follows this finite-length version of that program.

Several lines of work have produced quantum LDPC codes. Finite-geometry LDPC codes and quasi-cyclic (QC) LDPC codes have been studied as explicit ways to obtain regular sparse matrices \cite{kou_lin_fossorier_2001,kamiya_fossorier_2006,kamiya_affine_qc_2007}. In the quantum setting, QC-LDPC codes are often built by replacing nonzero entries of a base matrix by circulant permutation matrix (CPM) blocks, which preserves sparsity while allowing the CSS orthogonality constraints to be imposed at finite length \cite{hagiwara_imai_2007,camara_ollivier_tillich_2007,aly_2008}. In a different direction, hypergraph product and homological product constructions have improved the asymptotic rate and distance behavior of quantum LDPC code families \cite{tillich_zemor_2014,bravyi_hastings_2014,breuckmann_eberhardt_2021}. These product constructions are designed to obtain global scaling guarantees, and they do not directly address the finite-length target of this paper: a fixed length-$512$ girth-$8$ $(3,8)$-regular base.

Recent construction work has often emphasized high-rate quantum LDPC codes. A related line of quantum error-correction work by the second author and coauthors includes non-binary QC-LDPC decoding beyond the bounded-distance limit \cite{kasai_hagiwara_imai_sakaniwa_2012_qec_bdd}, codes approaching the coding-theoretic bound \cite{komoto_kasai_2025_coding_bound}, explicit girth-$12$ quantum QC-LDPC constructions \cite{komoto_kasai_2025_qc_girth12}, non-binary error-floor mitigation \cite{kasai_2025_error_floor_nonbinary}, joint belief-propagation decoding of quantum QC-LDPC codes \cite{komoto_kasai_2025_sharp_transitions}, randomized orthogonal sparse matrix pairs \cite{okada_kasai_2025_random_construction}, and high-rate active-latent constructions \cite{kasai_2026_orthogonality_barrier}.

In particular, the construction in \cite{kasai_2026_orthogonality_barrier} starts from parent matrices with many rows and separates them into active rows, which are used as stabilizer checks, and latent rows, which are not included in the code. If orthogonality is imposed on the full parent matrices, the latent rows are also orthogonal to the active rows, and deleted low-weight rows or their low-weight combinations can become logical operators. That construction therefore imposes CSS orthogonality only where it is needed between the $X$- and $Z$-active rows, while keeping mixed active-latent interactions nonzero. This prevents latent rows from automatically entering the dual of the active code and targets the high-rate regime, in particular design rate at least one half. As a concrete example, it constructs a girth-$8$ $(3,12)$-regular quantum LDPC code. Recent work on reconfigurable atom arrays also studies high-rate quantum error correction as an implementation target \cite{zhao_2026_logical_atom_arrays}.

In contrast, this paper does not optimize for the highest possible rate. We use a girth-$8$ $(3,8)$-regular LDPC code as the CSS constituent and study a finite-length construction with design rate one quarter. This choice keeps the stabilizer weight equal to $8$ while making the base matrix, the absence of short cycles, and the lifted-code decoding behavior directly checkable.

The purpose of this paper is to give one completely explicit length-$512$ base matrix pair and verify its finite-length properties. The underlying base code is not new as a CSS code: it is permutation-equivalent to the SPC(3) product CSS code in \cite{OstrevOrsucciLazaroMatuz2024}. What we add is a direct affine-coset description that makes the Tanner-graph structure transparent. More precisely, we directly define a CSS base matrix pair with row weight $8$, column weight $3$, and then prove within this description that both Tanner graphs have girth $8$. The construction can also be viewed in finite affine geometry: the qubits are the points of the affine space $\operatorname{AG}(9,2)$, and the checks are affine $3$-flats. In this paper, however, we use only the corresponding linear-algebraic definition by affine cosets. We prove orthogonality, regularity, and absence of short cycles, and record the exact parameters of the base code. We then extend this fixed base by CPM lifts and report decoding results.

The technical development below therefore keeps the object fixed: a length-$512$ $(3,8)$-regular base, its CPM lifts, and the decoder used in the experiments. More general affine-coset systems and other base lengths are natural extensions of the same viewpoint, but they are not needed for the finite-length verification reported here.

For the fixed $P=32$ lift used in the decoding experiment, the exact parameters are $[[16384,4142,32]]$, with $d_X=d_Z=32$. The exhaustive computation in \cite{zhang_p32_distance_2026} extends our earlier exclusion through weight $16$ to every even weight below $32$, using a coordinate constraint on the pairing graphs and a stabilizer-based minimality test. The explicit weight-$32$ logical representatives attain this lower bound. Appendix~\ref{sec:distance-enumeration} explains the reductions, the completeness argument, and the verification records.

Section~\ref{sec:base-construction} defines the base matrix pair, proves its regularity, orthogonality, and girth, and records the exact base-code parameters. Section~\ref{sec:lift-framework} describes the CPM lift and the orthogonality condition for the lifted matrices. Section~\ref{sec:decoding} establishes the exact distance of the experimental code and reports the decoder and frame error rate (FER).

\section[A girth-8 (3,8)-regular base construction]{A girth-$8$ $(3,8)$-regular base construction}
\label{sec:base-construction}
This section directly defines the length-$512$ base matrix pair used as the starting point for the lift. The only objects needed are $V=\F_2^9$, three subspaces $A,B,C$, and three further subspaces $D_1,D_2,D_3$ derived from them.
Equivalently, $V$ is the point set of the finite affine space $\operatorname{AG}(9,2)$, and the checks below are incidence vectors of affine $3$-flats. We keep the affine-coset notation because it gives the shortest proof of the CSS orthogonality and girth properties used in this paper.

\subsection{Definition of the base matrix pair}
Let $V=\F_2^9$. Decompose it as a direct sum of three $3$-dimensional subspaces,
\[
V=A\oplus B\oplus C .
\]
This means that every $\vect{x}\in V$ is written uniquely as
$\vect{x}=\vect{a}+\vect{b}+\vect{c}$ with $\vect{a}\in A$, $\vect{b}\in B$, and $\vect{c}\in C$.
Choose ordered bases
\[
A=\langle \vect{a}_1,\vect{a}_2,\vect{a}_3\rangle,\qquad
B=\langle \vect{b}_1,\vect{b}_2,\vect{b}_3\rangle,\qquad
C=\langle \vect{c}_1,\vect{c}_2,\vect{c}_3\rangle .
\]
Taking one basis vector from each block, define
\[
D_1=\langle \vect{a}_1,\vect{b}_1,\vect{c}_1\rangle,\qquad
D_2=\langle \vect{a}_2,\vect{b}_2,\vect{c}_2\rangle,\qquad
D_3=\langle \vect{a}_3,\vect{b}_3,\vect{c}_3\rangle .
\]
Then $D_1,D_2,D_3$ are also $3$-dimensional subspaces, and
\[
V=D_1\oplus D_2\oplus D_3 .
\]

The qubits are indexed by the elements of $V$. Hence the block length of the base construction is
\[
n=|V|=2^9=512 .
\]

For a linear subspace $U\subseteq V$, let $V/U$ denote the set of affine cosets of $U$:
\[
V/U=\{\vect{u}+U:\vect{u}\in V\}.
\]
For a subset $S\subset V$, write $\ind{S}\in\F_2^V$ for its incidence vector.

\begin{definition}[Base matrix pair]
On the $X$ side, the $X$-checks are all affine cosets of $A,B,C$. On the $Z$ side, the $Z$-checks are all affine cosets of $D_1,D_2,D_3$. For a linear subspace $U\subseteq V$, let $M(U)$ be the incidence matrix whose rows are indexed by $V/U$ and whose columns are indexed by $V$, with
\[
M(U)_{U',\vect{x}}=
\begin{cases}
1, & \vect{x}\in U',\\
0, & \vect{x}\notin U'
\end{cases}
\qquad (U'\in V/U,\ \vect{x}\in V).
\]
All subspaces $U$ used here have dimension $3$, so $|V/U|=2^{9-3}=64$; in particular,
\[
M(A),M(B),M(C),M(D_1),M(D_2),M(D_3)\in \F_2^{64\times512}.
\]
Define the base matrix pair by
\[
\HXB=
\begin{pmatrix}
M(A)\\
M(B)\\
M(C)
\end{pmatrix},
\qquad
\HZB=
\begin{pmatrix}
M(D_1)\\
M(D_2)\\
M(D_3)
\end{pmatrix}.
\]
\end{definition}

There is some arbitrariness in this representation. The subspaces $A,B,C$ need not be a specific triple; any three $3$-dimensional subspaces satisfying $V=A\oplus B\oplus C$ may be used. The ordered bases of $A,B,C$ are also arbitrary. After the bases are chosen, the same rule defines $D_i=\langle \vect{a}_i,\vect{b}_i,\vect{c}_i\rangle$. Any two such choices are related by an invertible linear transformation of $V$, and hence give equivalent CSS codes after relabeling qubits.

To display the matrices, one must also choose an ordering of the qubits and checks. The next example fixes the ordering used in Figure~\ref{fig:base-graph} and computes $\HXB,\HZB$ in that ordering.

\begin{example}[Base matrices in the standard basis]
\label{ex:standard-basis}
We describe the concrete instance used in Figure~\ref{fig:base-graph}. Let $\vect{e}_1,\ldots,\vect{e}_9$ be the standard basis of $\F_2^9$, and set
\[
(\vect{a}_1,\vect{a}_2,\vect{a}_3,\vect{b}_1,\vect{b}_2,\vect{b}_3,\vect{c}_1,\vect{c}_2,\vect{c}_3)
=
(\vect{e}_1,\vect{e}_2,\vect{e}_3,\vect{e}_4,\vect{e}_5,\vect{e}_6,\vect{e}_7,\vect{e}_8,\vect{e}_9).
\]
Then
\[
A=\langle \vect{e}_1,\vect{e}_2,\vect{e}_3\rangle,\quad
B=\langle \vect{e}_4,\vect{e}_5,\vect{e}_6\rangle,\quad
C=\langle \vect{e}_7,\vect{e}_8,\vect{e}_9\rangle,
\]
and
\[
D_1=\langle \vect{e}_1,\vect{e}_4,\vect{e}_7\rangle,\quad
D_2=\langle \vect{e}_2,\vect{e}_5,\vect{e}_8\rangle,\quad
D_3=\langle \vect{e}_3,\vect{e}_6,\vect{e}_9\rangle.
\]

The columns are ordered as follows. A qubit is a vector
\[
\vect{x}=\sum_{i=1}^3\alpha_i\vect{a}_i+\sum_{i=1}^3\beta_i\vect{b}_i+\sum_{i=1}^3\gamma_i\vect{c}_i\in V,
\]
and we interpret
\[
(\alpha_1,\alpha_2,\alpha_3,\beta_1,\beta_2,\beta_3,\gamma_1,\gamma_2,\gamma_3)
\]
as bits in this order.

In this setting, $M(A)$ is obtained explicitly as follows. An affine coset of $A$ is obtained by fixing the $B$- and $C$-coordinates and varying only the $A$-coordinates. Hence the rows are indexed by the $64$ cosets
\[
\sum_{i=1}^3\beta_i\vect{b}_i+\sum_{i=1}^3\gamma_i\vect{c}_i+A .
\]
For such a row and a column $\vect{x}$, the entry
\[
M(A)_{\left(\sum_i\beta_i\vect{b}_i+\sum_i\gamma_i\vect{c}_i+A\right),\,\vect{x}}
\]
is equal to $1$ if and only if the $B$- and $C$-coordinates of $\vect{x}$ are the fixed values
$(\beta_1,\beta_2,\beta_3,\gamma_1,\gamma_2,\gamma_3)$ of that row. Componentwise, if a column is written as
\[
\vect{x}_{\alpha,\beta',\gamma'}
=\sum_{i=1}^3\alpha_i\vect{a}_i
+\sum_{i=1}^3\beta'_i\vect{b}_i
+\sum_{i=1}^3\gamma'_i\vect{c}_i,
\]
then
\[
M(A)_{\left(\sum_i\beta_i\vect{b}_i+\sum_i\gamma_i\vect{c}_i+A\right),\,\vect{x}_{\alpha,\beta',\gamma'}}
=
\begin{cases}
1, & \beta'_i=\beta_i,\ \gamma'_i=\gamma_i\quad (i=1,2,3),\\
0, & \text{otherwise}.
\end{cases}
\]
Since $\HXB$ is defined by stacking $M(A),M(B),M(C)$ in this order, the top $64$ rows of $\HXB$ agree with $M(A)$ entry by entry. The $A$-coordinates $(\alpha_1,\alpha_2,\alpha_3)$ are then free, so each row has eight ones. With the displayed column order, the $A$-coordinates are the lowest-order bits, and these eight ones form one consecutive block of eight columns. As the row index runs through the $B,C$-coordinates in integer order, this block moves from left to right. This gives the top block $M(A)$ shown in Figure~\ref{fig:base-graph}.

The other matrices $M(B),M(C),M(D_1),M(D_2),M(D_3)$ are formed by the same incidence rule. Stacking them as
\[
\HXB=
\begin{pmatrix}
M(A)\\
M(B)\\
M(C)
\end{pmatrix},
\qquad
\HZB=
\begin{pmatrix}
M(D_1)\\
M(D_2)\\
M(D_3)
\end{pmatrix}
\]
gives the base matrix pair displayed in Figure~\ref{fig:base-graph}.
\end{example}

The exact base code is already known from \cite{OstrevOrsucciLazaroMatuz2024}.
What is new here is the affine-coset description that gives a direct proof of
the $(3,8)$-regular and girth-$8$ properties. For completeness, Appendix
~\ref{app:spc3-equivalence} records the definition of SPC(3) and the proof that
the present base matrix pair is permutation-equivalent to it.

The next proposition records that regularity and orthogonality follow directly from the definition.

\begin{proposition}
\label{prop:base-parameters}
Both $\HXB$ and $\HZB$ belong to $\F_2^{192\times512}$. Every row has weight $8$, every column has weight $3$, and
\[
\HXB{\HZB}^{\mathsf T}=0 .
\]
\end{proposition}

This proposition shows that the required code parameters are built into the construction. Neither regularity nor orthogonality is added later by search.

\begin{proof}
Each check is an affine coset of a $3$-dimensional subspace, so each row contains
\[
2^3=8
\]
elements. Each such subspace has
\[
2^{9-3}=64
\]
affine cosets, so each side has
\[
3\cdot64=192
\]
rows.

Fix $\vect{x}\in V$. Since the cosets of $A$ partition $V$, the point $\vect{x}$ lies in exactly one coset of $A$. The same holds for the cosets of $B$ and $C$. Hence each column of $\HXB$ has weight $3$. The same argument applies to $\HZB$ using the coset partitions of $D_1,D_2,D_3$.

For orthogonality, take an $X$-row $\vect{u}+S$ and a $Z$-row $\vect{v}+T$, where $S\in\{A,B,C\}$ and $T\in\{D_1,D_2,D_3\}$. Their inner product is
\[
\inner{\ind{\vect{u}+S}}{\ind{\vect{v}+T}}
=|(\vect{u}+S)\cap(\vect{v}+T)| \pmod 2 .
\]
For every such pair, $S\cap T$ is a one-dimensional subspace. Thus the intersection of the two affine cosets is either empty or a one-dimensional affine subspace, and has size $0$ or $2$. The inner product is therefore always zero, which proves
\[
\HXB{\HZB}^{\mathsf T}=0 .
\]
\end{proof}

We next verify that the Tanner graphs of the base matrix pair have no $4$-cycles and no $6$-cycles.

\begin{figure}[t]
  \centering
  \includegraphics[width=\textwidth]{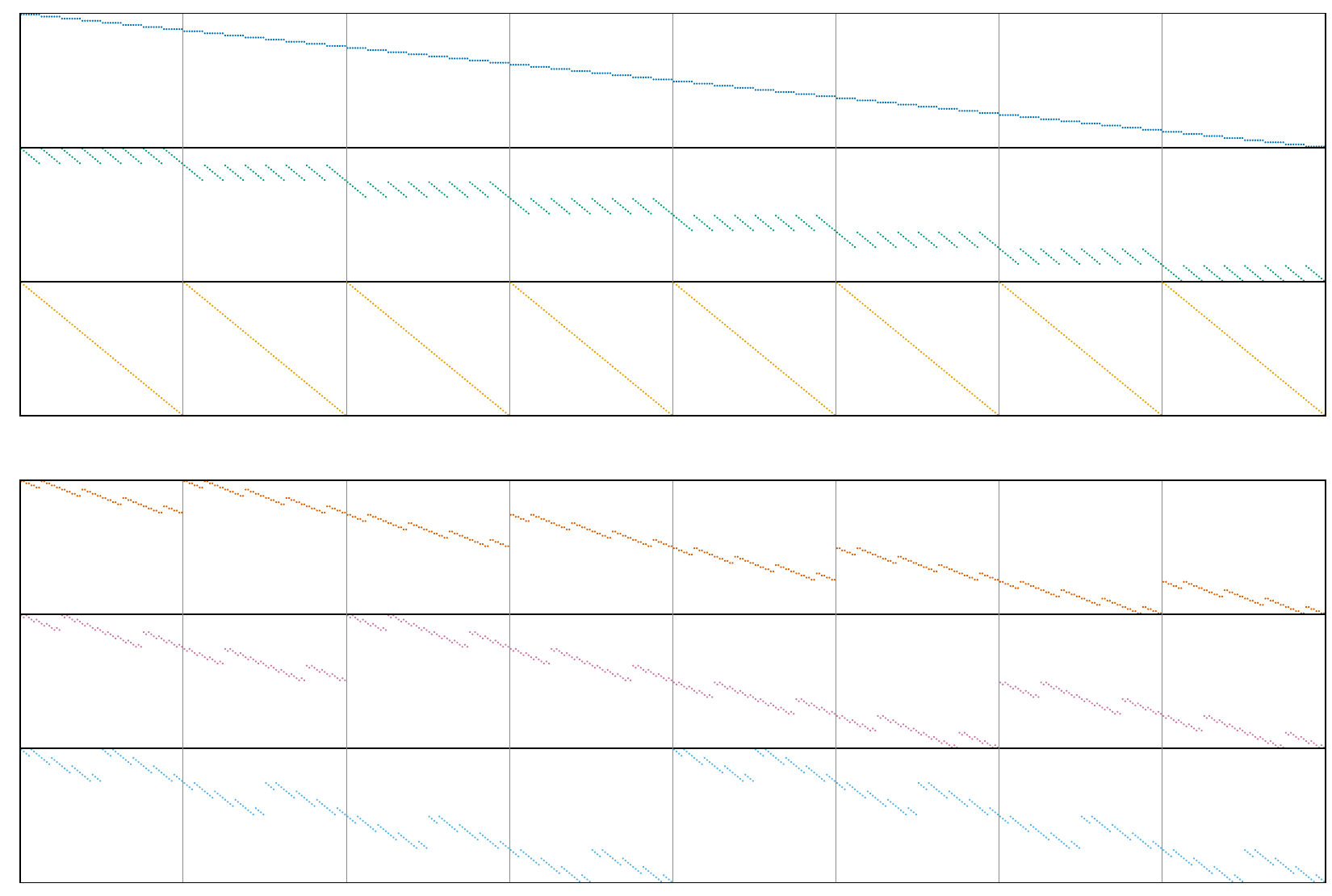}
  \caption{Visualization of the nonzero entries of the base matrices $\HXB$ and $\HZB$. The upper panel shows $\HXB$, and the lower panel shows $\HZB$. The rows of $\HXB$ are ordered as $A,B,C$, and the rows of $\HZB$ are ordered as $D_1,D_2,D_3$; the color indicates the row family. The columns are ordered by the usual integer order of the coordinates with respect to $\vect{a}_1,\vect{a}_2,\vect{a}_3,\vect{b}_1,\vect{b}_2,\vect{b}_3,\vect{c}_1,\vect{c}_2,\vect{c}_3$.}
  \label{fig:base-graph}
\end{figure}

\subsection{Local graph structure}
After defining the base matrix pair, we examine its local Tanner graph structure. For $4$-cycles we use the standard criterion that a $4$-cycle occurs exactly when two distinct rows on the same side share at least two columns \cite{tanner_1981}.

\begin{proposition}
\label{prop:no-4-cycle}
The Tanner graphs of $\HXB$ and $\HZB$ have no $4$-cycles.
\end{proposition}

This means that two checks on the same side never share more than one qubit.

\begin{proof}
A $4$-cycle exists if and only if two distinct rows on the same side share at least two columns.

Consider $\HXB$. Distinct cosets in the same coset family are disjoint. For two different families, for example $A$ and $B$, we have
\[
A\cap B=\{\zero\}.
\]
Hence $(\vect{u}+A)\cap(\vect{v}+B)$ is empty or a singleton. Indeed, if
\[
\vect{x},\vect{y}\in(\vect{u}+A)\cap(\vect{v}+B),
\]
then
\[
\vect{x}-\vect{y}\in A\cap B=\{\zero\},
\]
so $\vect{x}=\vect{y}$. The same argument applies to the pairs $A,C$ and $B,C$. Thus any two distinct rows of $\HXB$ share at most one column.

For $\HZB$, the same proof applies because
\[
D_1\cap D_2=D_2\cap D_3=D_3\cap D_1=\{\zero\}.
\]
\end{proof}

This result shows that the direct-sum decompositions by $A,B,C$ and by $D_1,D_2,D_3$ bound the number of common qubits between two checks by one.

\begin{proposition}
\label{prop:no-6-cycle}
The Tanner graphs of $\HXB$ and $\HZB$ have no $6$-cycles.
\end{proposition}

This proposition excludes a configuration in which three checks meet pairwise in three different qubits. Together with Proposition~\ref{prop:no-4-cycle}, it implies that the girth of the base Tanner graphs is at least $8$.

\begin{proof}
We prove the claim for $\HXB$. The proof for $\HZB$ is identical after replacing $(A,B,C)$ by $(D_1,D_2,D_3)$.

By Proposition~\ref{prop:no-4-cycle}, there are no $4$-cycles. If a $6$-cycle existed, there would be three rows
\[
R_A=\vect{u}+A,\qquad
R_B=\vect{v}+B,\qquad
R_C=\vect{w}+C
\]
whose pairwise intersections are three distinct points. Write
\[
\vect{x}\in R_A\cap R_B,\qquad
\vect{y}\in R_B\cap R_C,\qquad
\vect{z}\in R_C\cap R_A .
\]

Since $\vect{x},\vect{y}\in R_B$, we have $\vect{x}+\vect{y}\in B$. Similarly,
\[
\vect{y}+\vect{z}\in C,\qquad
\vect{z}+\vect{x}\in A .
\]
On the other hand,
\[
(\vect{x}+\vect{y})+(\vect{y}+\vect{z})+(\vect{z}+\vect{x})=\zero .
\]
Because
\[
V=A\oplus B\oplus C
\]
is a direct sum, the sum of one vector from each of $A,B,C$ can be zero only if all three components are zero. Hence
\[
\vect{x}+\vect{y}=\zero,\qquad
\vect{y}+\vect{z}=\zero,\qquad
\vect{z}+\vect{x}=\zero,
\]
and therefore
\[
\vect{x}=\vect{y}=\vect{z}.
\]
This contradicts the assumption that the three pairwise intersections are distinct points.
\end{proof}

Propositions~\ref{prop:no-4-cycle} and~\ref{prop:no-6-cycle} show that the base Tanner graphs have neither $4$-cycles nor $6$-cycles. Thus the lift in this paper is not used to remove short cycles from the base graph; it is used to increase the block length while preserving CSS orthogonality.

\begin{theorem}
\label{thm:base-summary}
The pair $(\HXB,\HZB)$ is a CSS base matrix pair over $\F_2^{192\times512}$. Every row has weight $8$, every column has weight $3$, and both Tanner graphs have girth $8$.
\end{theorem}

This theorem summarizes the base matrix pair used in Section~\ref{sec:lift-framework}.

\begin{proof}
Row weight, column weight, and the CSS orthogonality condition are given by Proposition~\ref{prop:base-parameters}. The absence of $4$-cycles and $6$-cycles follows from Propositions~\ref{prop:no-4-cycle} and~\ref{prop:no-6-cycle}.

It remains only to note that $8$-cycles exist. On the $\HX$ side, choose $0\neq a\in A$ and $0\neq b\in B$. Then
\[
A\;-\;0\;-\;B\;-\;b\;-\;(b+A)\;-\;(a+b)\;-\;(a+B)\;-\;a\;-\;A
\]
is an $8$-cycle. Similarly, on the $\HZ$ side, for $0\neq u\in D_1$ and $0\neq v\in D_2$,
\[
D_1\;-\;0\;-\;D_2\;-\;v\;-\;(v+D_1)\;-\;(u+v)\;-\;(u+D_2)\;-\;u\;-\;D_1
\]
is an $8$-cycle. Hence both Tanner graphs have girth exactly $8$.
\end{proof}

For reference, the CSS code defined by this base matrix pair has parameters
\[
[[512,174,8]] .
\]
This base code is permutation-equivalent to the SPC(3) product CSS code in
\cite{OstrevOrsucciLazaroMatuz2024}; see Appendix
~\ref{app:spc3-equivalence}. In that sense, the exact base-code
parameters are not new. What is proved in this section is the Tanner-graph
statement needed later in this paper: the present affine-coset description gives
a $(3,8)$-regular base pair whose two Tanner graphs have girth $8$. That girth
statement is established here from Propositions~\ref{prop:no-4-cycle}
and~\ref{prop:no-6-cycle}; it is not imported from \cite{OstrevOrsucciLazaroMatuz2024}.

Thus the base code has $k=174$ logical qubits and exact minimum distance $d=8$. The rank and distance calculation is not used in the lift construction below, so we record this fact without proof.
The parameters $[[n,k,d]]$ of this base code can be computed exactly from its tensor-product structure, but that computation is independent of the CPM lift and decoding experiments studied here.

\section{CPM lift}
\label{sec:lift-framework}
This section lifts the girth-$8$ $(3,8)$-regular base from Section~\ref{sec:base-construction} to longer codes while preserving CSS orthogonality. We use the standard notation for quasi-cyclic codes and replace each nonzero entry by a CPM. This is the usual circulant lifting framework in classical QC-LDPC and protograph-based constructions, and we follow the same block-replacement viewpoint here \cite{kamiya_fossorier_2006,MitchellSmarandacheCostello2014,hagiwara_imai_2007,camara_ollivier_tillich_2007}.

The base matrix pair is
\[
\HXB\in\F_2^{192\times512},\qquad
\HZB\in\F_2^{192\times512}.
\]
In a $P$-lift, each nonzero entry is replaced by a $P\times P$ CPM. Each side has
\[
512\cdot3=192\cdot8=1536
\]
nonzero entries, so a CPM lift is specified by $1536$ shift labels on each side. Rather than writing these labels as an unstructured list, we use edge-label functions on the support sets
\[
E_X:=\{(r,\vect{v}):\HXB(r,\vect{v})=1\},\qquad
E_Z:=\{(s,\vect{v}):\HZB(s,\vect{v})=1\}.
\]
Let
\[
x:E_X\to\mathbb Z_P,\qquad z:E_Z\to\mathbb Z_P
\]
be the two label functions. If $I_P(t)$ denotes the CPM corresponding to $t\in\mathbb Z_P$, define
\[
x_{r,\vect{v}}:=x(r,\vect{v}),\qquad
z_{s,\vect{v}}:=z(s,\vect{v})
\]
and replace the nonzero entries by
\[
H_X^{(P)}(r,\vect{v})=I_P(x_{r,\vect{v}}),\qquad
H_Z^{(P)}(s,\vect{v})=I_P(z_{s,\vect{v}}).
\]
This operation preserves row weight $8$ and column weight $3$, and the block length becomes $512P$.

The orthogonality condition is simple because of the two-point intersection structure of the base. In this base, the intersection of any $X$-row and any $Z$-row is either empty or consists of two points. If
\[
N_X(r)\cap N_Z(s)=\{\vect{v}_1,\vect{v}_2\},
\]
then the lifted rows are orthogonal over $\F_2$ whenever
\[
x_{r,\vect{v}_1}-z_{s,\vect{v}_1}
\equiv
x_{r,\vect{v}_2}-z_{s,\vect{v}_2}
\pmod P .
\]
All lifted matrix pairs used in this paper are CPM lifts satisfying this local congruence condition.

\section{Decoding experiments}
\label{sec:decoding}
The fixed CPM lift used in the decoding experiment has exact distance $32$. Its parameters and explicit logical representatives identify the code whose decoder, success criterion, and FER are reported below. The distance $d=8$ recorded in Section~\ref{sec:base-construction} is the exact distance of the base CSS code before lifting.

\subsection{Code parameters used in the decoding experiments}
We first record the basic parameters of the lifted codes used in the decoding experiments. After a $P$-lift, the number of physical qubits is
\[
n=512P,
\]
and the number of check rows on each side is $192P$. If the checks are counted as independent, the design dimension is
\[
512P-2\cdot 192P=128P,
\]
so the design rate is one quarter. In this paper we accept this design rate in order to keep column weight $3$, row weight $8$, girth $8$ at the base level, and stabilizer weight $8$ after lifting. Row reduction over $\F_2$ on the lifted parity-check matrices used in the experiment gives
\[
\operatorname{rank} H_X=\operatorname{rank} H_Z=192\cdot32-23=6121
\]
for the $P=32$ code used below. Hence the number of encoded qubits in this code is
\[
k=16384-2\cdot6121=4142.
\]
Therefore the code used in Figure~\ref{fig:fer} has
\[
[[16384,4142,32]].
\]
The distance computation in \cite{zhang_p32_distance_2026} extends the complete enumeration on these matrices through weight $30$. The lower-bound argument appears in Appendix~\ref{sec:distance-enumeration}; the matching upper bounds follow from the explicit representatives below.

Index a lifted qubit by $(v,t)$, where $v\in\{0,\ldots,511\}$ is the base-column index and $t\in\mathbb Z_{32}$ is the lift index. Let $K=\{0,8,16,24\}$. An $X$-type representative is supported on
\[
\bigcup_{(v,r)\in S_X}\{(v,r+k):k\in K\},
\]
\[
S_X=\{(64,0),(68,7),(72,1),(76,7),(80,5),(84,5),(88,6),(92,6)\},
\]
and a $Z$-type representative is supported on
\[
\bigcup_{(v,r)\in S_Z}\{(v,r+k):k\in K\},
\]
\[
\begin{aligned}S_Z=\{&(134,0),(135,5),(142,4),(143,7),\\
&(262,0),(263,2),(270,6),(271,5)\}.\end{aligned}
\]
with all lift indices reduced modulo $32$. Direct binary computation gives zero syndrome for each representative. Moreover,
\[
\begin{aligned}
\operatorname{rank}H_X&=6121,
&\operatorname{rank}\begin{bmatrix}H_X\\ \vect{x}\end{bmatrix}&=6122,\\
\operatorname{rank}H_Z&=6121,
&\operatorname{rank}\begin{bmatrix}H_Z\\ \vect{z}\end{bmatrix}&=6122,
\end{aligned}
\]
where $\vect{x}$ and $\vect{z}$ are the displayed $X$- and $Z$-type representatives, respectively. Thus neither vector is a stabilizer, and
\[
d_X\le32,\qquad d_Z\le32,\qquad d\le32.
\]

\begin{proposition}[Exact distance of the fixed experimental code]
\label{prop:exact-distance}
The $P=32$ lifted matrices used in Figure~\ref{fig:fer} define a $[[16384,4142,32]]$ CSS code with $d_X=d_Z=32$.
\end{proposition}
\begin{proof}
The ranks above give $k=4142$. The completed searches exclude nontrivial pure logical representatives of every even weight through $30$ on both sides; Appendix~\ref{sec:distance-enumeration} gives their coverage and the reductions used in \cite{zhang_p32_distance_2026}. Odd weights are excluded because every column of either check matrix has odd weight. The displayed non-stabilizer representatives have weight $32$, so $d_X=d_Z=32$. A nontrivial mixed Pauli logical operator has at least one nontrivial pure component, whose support is contained in its support. Thus $d=\min\{d_X,d_Z\}=32$.
\end{proof}

\subsection{Decoder and noise model}
All simulations in this section use the code-capacity depolarizing model. In each trial, a Pauli error is sampled on the data qubits, the syndrome is computed perfectly from that error, and the decoder is given this single syndrome. We do not simulate the phenomenological model: there are no measurement errors, no noisy syndrome bits, and no repeated syndrome-extraction rounds.

Write a Pauli error as $(\vect{x},\vect{z})$. The observed syndromes are
\[
  \vect{s}_X=H_X\vect{z},\qquad
  \vect{s}_Z=H_Z\vect{x}.
\]
Let the decoder output be $(\tilde{\vect{x}},\tilde{\vect{z}})$. We say that decoding succeeds if
\[
\begin{gathered}
  H_X\tilde{\vect{z}}=\vect{s}_X,\qquad
  H_Z\tilde{\vect{x}}=\vect{s}_Z,\\
  \tilde{\vect{x}}+\vect{x}\in\mathrm{row}(H_X),\qquad
  \tilde{\vect{z}}+\vect{z}\in\mathrm{row}(H_Z).
\end{gathered}
\tag{SC}
\]
The first line says that the output satisfies the observed syndrome. The second line says that the difference from the true error is a stabilizer. Indeed, if the second line holds, then $\tilde{\vect{x}}+\vect{x}$ is a product of $X$-stabilizers and $\tilde{\vect{z}}+\vect{z}$ is a product of $Z$-stabilizers. Stabilizers act trivially on the code space, so such an output has the same logical class as the true error. The CSS orthogonality conditions also imply $\mathrm{row}(H_X)\subseteq\ker H_Z$ and $\mathrm{row}(H_Z)\subseteq\ker H_X$, so the row-space conditions are consistent with the syndrome equations.

The decoder used here is an experimental decoder. It first runs BP and then repairs the BP estimate by post-processing. BP itself is used as a standard iterative decoder, and we do not describe its update equations. The post-processing is based on OSD for classical soft-decision decoding \cite{fossorier_lin_1995} and BP+OSD decoding for quantum LDPC codes \cite{roffe_white_burton_campbell_2020}. It repairs the remaining unsatisfied syndrome on sets of low-reliability qubits. Let the BP output be $(\hat{\vect{x}},\hat{\vect{z}})$. The decoder receives only $H_X,H_Z$, the observed syndrome, and the error rate $p$. The true noise $(\vect{x},\vect{z})$ generated in the simulation is not used by the decoder. It is used only after decoding to test condition (SC) and classify logical errors. The implementation and the $P=32$ CPM-lifted code data used for the reported experiments are available in a public software repository \cite{kasai_quantum_ldpc_decoder_2026}.

In the implementation, after BP and after each post-processing step, the current estimate is denoted by $(\hat{\vect{x}},\hat{\vect{z}})$ and condition (SC) is tested in the following order. First we check
\[
  H_X\hat{\vect{z}}=\vect{s}_X,\qquad
  H_Z\hat{\vect{x}}=\vect{s}_Z .
\]
If this fails, some syndrome remains and the instance is passed to post-processing. If the syndrome equations hold, we then test whether
\[
  \hat{\vect{x}}+\vect{x},\qquad
  \hat{\vect{z}}+\vect{z}
\]
belong to the row spaces of $H_X$ and $H_Z$, respectively. If they do, condition (SC) holds and decoding is counted as successful. Otherwise, the output is counted as a nontrivial logical error. The FER in Figure~\ref{fig:fer} counts both nonconvergent failures with residual syndrome and syndrome-satisfying logical errors.

Post-processing is called only when a residual syndrome remains. It tries the following steps in order.

\begin{enumerate}
\item \textbf{Residual syndrome setup.}
To repair the $X$ component, define the residual syndrome
\[
  \vect{r}_X:=\vect{s}_Z+H_Z\hat{\vect{x}}.
\]
To repair the $Z$ component, define
\[
  \vect{r}_Z:=\vect{s}_X+H_X\hat{\vect{z}}.
\]
For the $X$ component, the goal is to find an additional correction $\vect{\delta}$ satisfying
\[
  H_Z\vect{\delta}=\vect{r}_X
\]
and then update $\hat{\vect{x}}$ to $\hat{\vect{x}}+\vect{\delta}$. The $Z$ component is treated analogously by solving $H_X\vect{\delta}=\vect{r}_Z$.

The candidate correction is chosen using the reliabilities produced by BP. For each qubit $v$, let $c_v^X$ be the cost of flipping the $X$ component and $c_v^Z$ the cost of flipping the $Z$ component. Smaller cost means that BP regards the flip as more plausible. For $X$-component repair, among candidates satisfying $H_Z\vect{\delta}=\vect{r}_X$, we choose one minimizing
\[
  C_X(\vect{\delta})=\sum_{v:\delta_v=1}c_v^X .
\]
For $Z$-component repair, we minimize
\[
  C_Z(\vect{\delta})=\sum_{v:\delta_v=1}c_v^Z .
\]
Below, for one side, we write the check matrix as $H$, the residual syndrome as $\vect{r}$, and the current estimate as $\hat{\vect{u}}$. Order the qubits from least reliable to most reliable as
\[
  v_1,v_2,\ldots,v_n,
\]
where the first positions are the least reliable qubits. For an integer $m$, define
\[
  K_m:=\{v_1,\ldots,v_m\}.
\]
Also, $H_{*,K_m}$ denotes the submatrix of $H$ obtained by keeping all check rows and only the columns in $K_m$. The symbol $*$ means that no check rows are removed. The vector $\vect{\delta}_{K_m}$ is the correction restricted to $K_m$, and it is extended by zero outside $K_m$.

\item \textbf{$K_m$ bisection OSD.}
On the candidate set $K_m$, we solve
\[
  H_{*,K_m}\vect{\delta}_{K_m}=\vect{r}
\]
over $\F_2$. This linear system asks whether the remaining syndrome can be removed by flipping only the unreliable qubits in $K_m$. The implementation uses bit-packed Gaussian elimination and always recomputes the syndrome after applying a correction.

The size of $K_m$ is chosen by bisection. Such post-processing, where the candidate-set size is adjusted by solvability, is used as a low-complexity BP post-processing method \cite{kasai_2026_orthogonality_barrier}. We first find the smallest solvable size
\[
  m_{\mathrm{sol}}:=\min\{m:\ H_{*,K_m}\vect{\delta}_{K_m}=\vect{r}\ \text{is solvable}\}.
\]
If the solution is unique, that solution is used as the candidate correction. If the solution is not unique, syndrome information alone does not determine which correction to choose. However, if the difference between any two solutions is always a stabilizer, then all solutions have the same logical action on the code space.

More explicitly, write the solution space as
\[
  \vect{\delta}_0+\mathcal N,\qquad
  \mathcal N:=\{\vect{\eta}:H_{*,K_m}\vect{\eta}=0\}.
\]
For any two solutions $\vect{\delta}_1,\vect{\delta}_2$,
\[
  \vect{\delta}_1+\vect{\delta}_2\in\mathcal N .
\]
When repairing the $X$ component, the vector $\vect{\delta}$ is an $X$-type additional correction. If every $\vect{\eta}\in\mathcal N$ lies in $\mathrm{row}(H_X)$, then any two corrections differ only by an $X$-stabilizer. They can therefore be treated as the same logical correction. Similarly, when repairing the $Z$ component, we test whether $\mathcal N\subseteq\mathrm{row}(H_Z)$. If this condition holds, the solution is unique modulo stabilizers, and we choose from the solution space the vector with minimum flip cost $C_X$ or $C_Z$. The minimized quantity is not Hamming weight but the BP reliability-based cost. A low-Hamming-weight correction may flip highly reliable qubits and can be less likely than a slightly higher-weight correction on unreliable qubits.

The linear equation guarantees only that the corrected estimate satisfies the syndrome. Decoding success still requires the row-space part of condition (SC), namely that the difference from the true error is a stabilizer. The stabilizer-modulo test is used only to ensure that, when OSD has multiple solutions, those solutions represent the same logical class.

\item \textbf{Local OSD.}
If $K_m$ bisection OSD does not repair the residual syndrome, we try a smaller OSD restricted to the neighborhood of the unsatisfied checks. The candidate set consists of the checks with nonzero residual syndrome, their adjacent qubits, and, when needed, checks one step farther away. On this candidate set we solve the same linear equation and accept only corrections that remove the syndrome. This step is intended for local residuals and avoids building a large candidate set when the remaining syndrome is already localized.

\item \textbf{Restricted-candidate fallback.}
If local OSD is still insufficient, we do not expand the candidate set to all qubits. Large candidate sets make Gaussian elimination and solution-space tests expensive, and they also tend to create many nonunique solutions. The fallback used here is a restricted auxiliary repair: it is called only when the number of remaining unsatisfied checks is small, and it searches only predetermined small supports.

We try three types of supports. First, we connect nearby unsatisfied checks on the base graph and close the corresponding neighborhood in the lift, producing path-closure candidates. This targets residuals that remain along short path-like structures. Second, we build candidates that jointly treat two unsatisfied checks coming from the same base row or the same base check family. This handles residuals that may appear separated in the lift but share a simple local structure in the base. Third, if the residual matches a small trapping-set template, we search only on that template support.

This fallback is not maximum-likelihood decoding. It is a local repair designed to keep the computation bounded. In every case, the true noise is not used; the candidate tests use only the observed syndrome and the residual syndrome computed from the current estimate.

\item \textbf{Re-BP and joint repair.}
If only one side is improved, we run a short BP pass once more using the updated estimate as the initial state. This uses the fact that removing the syndrome on one side can change the reliability order and the local residual structure on the other side. If unsatisfied checks remain on both the $X$ and $Z$ sides, we try a small joint repair. The joint repair forms a candidate set from the union of the neighborhoods of both residual syndromes and tests updates of the $X$ and $Z$ components together. A candidate is accepted only if it satisfies both syndromes simultaneously.
\end{enumerate}

This decoder is not a maximum-likelihood decoder. It repairs the BP output using low-reliability sets and local candidates. Thus FER is not an exact distance estimate. However, if the final output satisfies the syndrome and the residual is not equivalent to a stabilizer, then that residual is evidence of a logical error and gives a distance upper bound.

\subsection{FER}
Figure~\ref{fig:fer} plots the FER for the decoder described above under the code-capacity depolarizing model. Each point is obtained by decoding from the observed perfect syndrome and the error rate $p$ only; the true noise is not provided to the decoder.

At the leftmost point, $p=0.085$, we observed $23$ final failures in $2{,}310{,}469{,}200$ trials. One of them satisfied the syndrome after post-processing, but the difference from the true error was not in the stabilizer row space. In that event, the $X$ residual was zero and the $Z$ residual had weight $40$. This is an independently observed logical residual, although it does not improve the structural upper bound $d\le32$ above. The remaining $22$ failures had residual syndrome and do not by themselves give a distance upper bound. Many of them are local nonconvergent instances with small residual syndrome weight and may be corrected by stronger post-processing. In the FER reported here, they are still counted as final failures of the decoder used in this paper.

The line labeled \emph{DE reference: (3,8)-regular} shows, as a reference line, the asymptotic threshold obtained from DE for BP on a locally tree-like sparse-graph ensemble with column weight $3$, row weight $8$, and the four-state prior distribution of the code-capacity depolarizing channel \cite{richardson_urbanke_2001}. The random ensemble used in DE is a reference ensemble and need not satisfy the finite-length CSS orthogonality constraints imposed in this paper. We plot the threshold as a black dash-dotted line at $p_{\mathrm{DE}}\simeq0.1009$. This line is not a finite-length performance guarantee for the lifted code; it is a reference value for interpreting the FER curve.

The comparison should be read with this difference in randomness in mind. A fully random sparse-graph LDPC ensemble has independent local choices up to the degree constraints, while our code is obtained from one fixed CSS-orthogonal base and CPM shift labels that must satisfy the lift orthogonality conditions. Thus the CPM construction is much more constrained than the ensemble used in DE. Nevertheless, the observed FER curve is close to the DE reference on the scale shown in Figure~\ref{fig:fer}. This suggests that, for the present finite length and decoder, the randomness supplied by CPM shift labels may already be sufficient to obtain much of the behavior predicted for random $(3,8)$-regular sparse graphs. This observation does not prove an asymptotic threshold and does not remove the CSS constraints, but it gives a reason to compare CPM lifts carefully before moving to more general permutation lifts.

\begin{figure}[t]
  \centering
  \includegraphics[width=\textwidth]{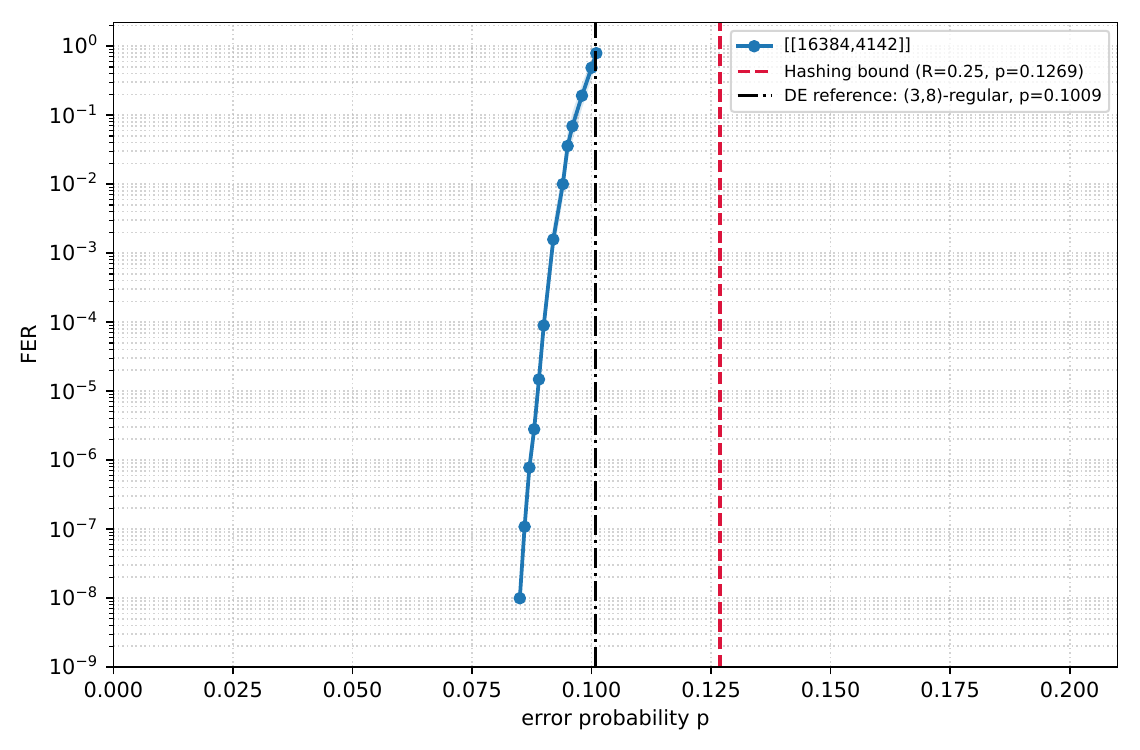}
  \caption{FER for the $P=32$ CPM-lifted code $[[16384,4142,32]]$ under the code-capacity depolarizing model. The curve shows the BP decoder with post-processing, and the shaded band gives the $95\%$ Wilson confidence interval. The red dashed line is the hashing bound for rate $1/4$, and the black dash-dotted line is a $(3,8)$-regular DE reference for a random sparse-graph ensemble, not a threshold for a CSS-orthogonal ensemble.}
  \label{fig:fer}
\end{figure}

\section{Future work}
This paper fixes one length-$512$ girth-$8$ $(3,8)$-regular CSS base matrix pair and studies its CPM lifts. The construction and the decoding experiment are finite-length and explicit. The main open questions are how the distance behaves after lifting, how far the decoder can be improved, and how much freedom remains when the CPM restriction is relaxed.

The first direction is to determine how distance depends on the lift factor and shift labels. The exact value $d_X=d_Z=32$ is now established for the fixed $P=32$ instance. Applying the coordinate-compatible support enumeration of Appendix~\ref{sec:distance-enumeration} to other lifts would permit systematic comparisons and help identify which shift-label choices increase the distance.

The second direction is decoder improvement. The failures observed in this paper have two different meanings. If a residual syndrome remains, the decoder has not yet found any correction consistent with the observed syndrome; this calls for stronger local or global syndrome repair. If the syndrome is satisfied but the difference from the true error is not a stabilizer, then the decoder has chosen the wrong logical class among syndrome-compatible corrections; this calls for better selection among nonunique OSD solutions, taking equivalence modulo the stabilizer row space into account. The next step is to analyze the residual-error structures behind these two failure types, classify the local residual syndromes and wrong logical-class choices that occur, and use this classification to design new post-processing decoders. Such a decoder should improve both the construction of candidate sets and the selection of representatives modulo the stabilizer row space.

The third direction is to go beyond CPM lifts. CPM labels form an additive cyclic group, so the orthogonality and short-cycle constraints can be written as congruence conditions. Replacing this cyclic group by affine permutation matrix lifts or by larger subgroups of $S_P$ would introduce more randomness into the lift and may make it easier to avoid unfavorable lifted structures. The cost is that the corresponding constraints become noncommutative and harder to check explicitly. At the same time, the CPM lift studied here already gives finite-length FER behavior close to the DE reference for a random sparse-graph LDPC ensemble. It may therefore be that CPM lifts already provide enough randomness for the finite-length regime considered in this paper. A useful next step is to determine when the additional randomness of more general permutation lifts gives a real decoding or distance advantage.

The fourth direction is to vary the base code. Since the base used in this paper is permutation-equivalent to SPC(3), this direction should be understood as a systematic study of nearby SPC-product and affine-coset bases, not as a claim that SPC-type generalizations are new. One possibility is to start from a shorter base graph of girth $6$ and choose a lift whose Tanner graph has girth at least $8$, aiming at shorter finite-length codes. Another is to use the affine-coset viewpoint to compare different degrees and finite-affine-geometric descriptions, and to identify which choices preserve the regularity, orthogonality, and lift constraints used here. These extensions would clarify which properties are specific to the present length-$512$ SPC(3) base and which persist in a broader class of high-girth regular quantum LDPC codes.

\section{Conclusion}
We gave an explicit affine-coset description of a length-$512$ CSS base matrix pair over $V=\F_2^9$. The underlying base CSS code is permutation-equivalent to the known SPC(3) product CSS code. The two matrices have $192$ rows on each side, row weight $8$, column weight $3$, satisfy CSS orthogonality, and have Tanner graphs of girth $8$ on both sides. The associated base CSS code has parameters $[[512,174,8]]$; in this paper we record these exact parameters but do not include their tensor-product proof.

We then lifted this fixed base by replacing each nonzero entry with a CPM. The two-point intersection structure of the base gives a local congruence condition on the shift labels, and this condition preserves CSS orthogonality after lifting. Thus the lifted codes keep stabilizer weight $8$ and block length $512P$ while remaining directly specified by finite shift-label data.

For the fixed $P=32$ code, the exhaustive distance computation in \cite{zhang_p32_distance_2026} and the explicit weight-$32$ logical representatives establish $d_X=d_Z=32$, giving parameters $[[16384,4142,32]]$. We also reported FER data under the code-capacity depolarizing model using BP with post-processing. The decoder uses the observed syndrome and the error rate, but not the true error. At $p=0.085$, the FER is at the $10^{-8}$ scale, and one syndrome-satisfying logical residual of weight $40$ was observed independently. The resulting finite-length construction has a computer-assisted exact distance and reproducible decoding data. Distance comparisons across other lifts and stronger post-processing remain open directions.

\section*{Acknowledgments}
We thank Anthony Leverrier for drawing our attention to the SPC(3) product code of Ostrev et al. and its connection to our base code. We also thank Shi-Xin Zhang for communicating the exact-distance computation for the fixed $P=32$ code and for making the programs, search records, and reproduction instructions publicly available.
\appendix
\section{SPC(3) and the base-code equivalence}
\label{app:spc3-equivalence}

This appendix records the definition of SPC(3) from
\cite{OstrevOrsucciLazaroMatuz2024} and the proof that the base construction in
Section~\ref{sec:base-construction} is permutation-equivalent to it.

\begin{definition}[SPC(3)]
In \cite{OstrevOrsucciLazaroMatuz2024}, SPC(3) is the case $(D,s)=(3,1)$ of
their CSPC$(D,s)$ family, that is, the $3$-fold product CSS code built from the
length-$2$ single-parity-check code. Its qubits are indexed by
\[
(\alpha,\beta,\gamma)\in \F_2^3\times\F_2^3\times\F_2^3\cong \F_2^9,
\]
where $\alpha=(\alpha_1,\alpha_2,\alpha_3)$, $\beta=(\beta_1,\beta_2,\beta_3)$,
and $\gamma=(\gamma_1,\gamma_2,\gamma_3)$.

The three $X$-check families are the $8$-point sets obtained by fixing two of
the three blocks $(\alpha,\beta,\gamma)$ and varying the remaining block
freely. The three $Z$-check families are the $8$-point sets obtained by fixing
all coordinates except one aligned coordinate position across the three blocks
and varying those three coordinates freely.
\end{definition}

\begin{theorem}[Equivalence to SPC(3)]
\label{thm:spc3-identification}
The base matrix pair defined in Section~\ref{sec:base-construction} is
permutation-equivalent to the SPC(3) product CSS code of
\cite{OstrevOrsucciLazaroMatuz2024}. More precisely, in the standard-basis
realization of Example~\ref{ex:standard-basis}, the $X$-check families
$A,B,C$ coincide with the three $X$-check families of SPC(3), and the
$Z$-check families $D_1,D_2,D_3$ coincide with the three $Z$-check families of
SPC(3). Hence the resulting CSS code is the SPC(3) code up to the ordering of
qubits and checks. For an arbitrary choice of $A,B,C$ and ordered bases, the
resulting base code is obtained from this standard-basis realization by a qubit
permutation.
\end{theorem}

\begin{proof}
The $X$-checks of SPC(3) are the $8$-point sets obtained by fixing two of the
three blocks and varying the remaining block freely. In the standard-basis
realization of Example~\ref{ex:standard-basis}, this gives exactly the affine
cosets of $A,B,C$: the cosets of $A$ are obtained by fixing $(\beta,\gamma)$
and varying $\alpha$, the cosets of $B$ are obtained by fixing $(\alpha,\gamma)$
and varying $\beta$, and the cosets of $C$ are obtained by fixing
$(\alpha,\beta)$ and varying $\gamma$. Therefore the three $X$-check families
agree.

The $Z$-checks of SPC(3) are the $8$-point sets obtained by fixing all
coordinates except one aligned coordinate position across the three blocks and
varying those three coordinates freely. In the same standard basis, these are
exactly the affine cosets of
\[
D_1=\langle \vect{e}_1,\vect{e}_4,\vect{e}_7\rangle,\qquad
D_2=\langle \vect{e}_2,\vect{e}_5,\vect{e}_8\rangle,\qquad
D_3=\langle \vect{e}_3,\vect{e}_6,\vect{e}_9\rangle .
\]
Indeed, a coset of $D_1$ is obtained by fixing
$(\alpha_2,\alpha_3,\beta_2,\beta_3,\gamma_2,\gamma_3)$ and varying
$(\alpha_1,\beta_1,\gamma_1)$; the cases of $D_2$ and $D_3$ are identical with
the coordinate position shifted. Therefore the three $Z$-check families also
agree.

Thus, in the standard-basis realization, our construction and SPC(3) have the
same row supports on both sides. Changing the ordering of qubits or checks only
applies permutation matrices to the columns or rows, so the resulting CSS code
is the same up to ordering conventions.

For a general choice of $A,B,C$ and ordered bases, let
\[
T\in \mathrm{GL}(V)
\]
be the unique linear isomorphism satisfying
\[
T(\vect{e}_1)=\vect{a}_1,\ \ldots,\ T(\vect{e}_9)=\vect{c}_3 .
\]
Then $T$ maps the standard coordinate subspaces to $A,B,C$, and therefore maps
the standard aligned subspaces to $D_1,D_2,D_3$. Since $T$ is a bijection on the
$512$ points of $V$, it induces a permutation of qubit coordinates. Hence the
general base code is permutation-equivalent to the standard-basis realization,
and therefore to SPC(3).
\end{proof}

\section{Computer-Assisted Exact Distance}\label{sec:distance-enumeration}
The exact distance of the fixed $P=32$ code follows by excluding minimum logical supports of weight below $32$ and verifying the explicit weight-$32$ representatives in Proposition~\ref{prop:exact-distance}. Our earlier exhaustive search excluded weights through $16$. The computation in \cite{zhang_p32_distance_2026} extends that exclusion through weight $30$ using the coordinate structure of the base and a stabilizer-based minimality test. The argument below explains why these reductions preserve every possible minimum logical support.

\subsection{Pairing graphs and minimum logical supports}
For a binary CSS pair, the pure distances are the minimum weights in
$\ker H_Z\setminus\operatorname{row}(H_X)$ and
$\ker H_X\setminus\operatorname{row}(H_Z)$, respectively.
The computation tests membership in the corresponding stabilizer row space by exact binary elimination.
Since every column of either matrix has odd weight three, summing all check equations shows that every kernel vector has even weight.

Each column meets exactly one check in each of three check-row groups, which we label by three colors.
A zero-syndrome support meets each check an even number of times.
Pair its variables within each check; the three groups give three colored perfect matchings on the support.
The absence of Tanner 4- and 6-cycles implies that their union is a simple triangle-free cubic graph.
Pairings at checks meeting four or more support variables are included.
A minimum-weight nontrivial logical support gives a connected graph under any such pairing: every connected component is itself a kernel support, and if all components were stabilizers then so would be their sum. Otherwise, a proper nontrivial component would contradict minimality.

The original enumeration retained one representative of each connected, simple, triangle-free graph up to color-preserving isomorphism. At weights $6,8,10,12,14,16$, its complete lists contained $1,10,22,226,1838,25375$ patterns, respectively. All were embedded in the fixed lifted matrices, and every completed embedding was tested for stabilizer membership. No nontrivial logical support was found.

\subsection{Coordinate constraint and stabilizer reduction}
For either check matrix separately, the $512$ base positions admit coordinates in $\{0,\ldots,7\}^3$ such that a check of color $c$ varies coordinate $c$ over its eight values and fixes the other two. This follows from the direct-sum descriptions of the three check families in Appendix~\ref{app:spc3-equivalence}. The input verification also recovers these coordinates from the fixed matrices: deleting one color leaves eight connected components, their three labels distinguish all $512$ positions, and each check has the stated coordinate-line property.

Two distinct lifted qubits on the same lifted check project to distinct base positions, because each nonzero base entry is replaced by a permutation matrix. Thus an edge of color $c$ in a pairing graph changes base coordinate $c$, whereas a path avoiding color $c$ preserves it. Consequently, the endpoints of a color-$c$ edge cannot lie in the same component after that color is deleted. This necessary condition can be checked before embedding the graph \cite{zhang_p32_distance_2026}.

The first two matchings form disjoint alternating even cycles of length at least four. The third matching must join different cycles. Each such cycle therefore has length at most half the support size. Enumerating all partitions into these cycle lengths and all admissible third matchings covers every coordinate-compatible graph. When adding a third-color edge, components for each of the other deleted colors are updated, and a branch is rejected only if it violates the coordinate constraint. Exchanges of equal-length untouched cycles and color-preserving symmetries within an untouched cycle remove equivalent choices. A canonical colored breadth-first traversal then retains one representative per connected color-preserving isomorphism class. At weight $16$, the coordinate constraint reduces the catalogue from $25375$ to $26$ graphs without changing the embedding counts on either side.

For the higher weight levels, a second reduction uses the weight-$8$ same-type stabilizer rows. Let $S$ be a minimum nontrivial logical support and let $T$ be such a stabilizer support. Adding the stabilizer preserves the nontrivial logical class and gives
\[
  |S\mathbin{\triangle}T|=|S|+8-2|S\cap T|.
\]
Hence $|S\cap T|\le4$. A partial support already meeting $T$ in at least five positions cannot extend to a minimum logical support and may be discarded. This test is applied at weights $24$ through $30$ in \cite{zhang_p32_distance_2026}. It is justified by minimality, including when the lighter representative obtained by adding $T$ is disconnected.

\subsection{Embedding, coverage, and verification}
For each retained graph, the embedding search assigns distinct qubits in the full lifted matrix. Along a color-$c$ edge, it tries the qubits incident to the same color-$c$ check as the assigned endpoint. Previously assigned neighbors impose check-equality constraints; a known coordinate shared along a path avoiding color $c$ further restricts the choice. These tests express necessary conditions for an embedding, and every compatible choice is explored. Completed supports have zero syndrome by their paired incidences and are tested by exact binary elimination against the same-type stabilizer row space.

The simultaneous cyclic shift of all lift indices preserves both parity-check matrices up to row permutations, and hence preserves their kernels and stabilizer row spaces. It suffices to assign the root to lift index zero for each of the $512$ base columns. No additional identification of base columns is used. Only fully completed weight levels on both sides contribute to the exclusion.

The published search records in \cite{zhang_p32_distance_2026} document complete searches through every even weight up to $30$ in both sectors, with no nontrivial logical representative found. Detailed catalogue sizes, embedding counts, and completion records are available in the cited repository at the fixed revision specified in the reference. We reran the enumeration through weight $30$ on the fixed input matrices and compared the catalogue sizes, completed-pattern counts, and embedding counts with those records. The reconstructed matrices agree with the hashes in the earlier ancillary certificate, and binary elimination gives rank $6121$ on each side. We also verified the displayed weight-$32$ representatives by zero syndromes and the rank increases in Section~\ref{sec:decoding}. The weight-$32$ representatives in the ancillary certificate are checked using zero syndromes and odd overlap with separating vectors in the opposite kernel.

Additional controls recover the known weight-$32$ logical supports with and without the overlap test. Through weight $22$, a separate complete-matching enumeration reproduces the graph catalogues; the embedding counts through weight $16$ agree with the earlier unrestricted catalogues, and at weights $18$ through $22$ they agree with the earlier embedding implementation applied to the reduced catalogues. The exclusions at weights $24$ through $30$ use the enumeration in \cite{zhang_p32_distance_2026}.

If a logical operator of weight below $32$ existed, a minimum nontrivial pure logical support would have even weight at most $30$ and survive all reductions above. The completed searches exclude this possibility. Together with the weight-$32$ witnesses, they establish $d_X=d_Z=32$. The ancillary files provide rerun records, hashes, and reproduction instructions for the cited fixed software revision.

\bibliographystyle{IEEEtran}
\bibliography{refs}

\end{document}